\documentclass[a4paper,12pt,oneside,dvipsnames]{article} 

\usepackage{amsmath,amsfonts,amsthm,amssymb}
\usepackage{xcolor}
\usepackage{graphicx}
\usepackage{geometry}
\usepackage{enumerate}
\usepackage{array}
\usepackage{bbm}
\usepackage{bm}
\usepackage{times}

\usepackage{natbib}	
\setcitestyle{round}

\usepackage{hyperref}
\hypersetup{colorlinks=true, citecolor=RoyalBlue, linkcolor=PineGreen, linktoc=page, urlcolor=PineGreen}

\theoremstyle{plain}
\newtheorem{theorem}{Theorem}[section] 
\newtheorem{lem}[theorem]{Lemma}
\newtheorem{prop}[theorem]{Proposition}

\newtheorem{ass}{Assumption}[section]

\theoremstyle{remark}
\newtheorem{rem}[theorem]{Remark}

\theoremstyle{definition}

\newcommand{\R}{\mathbb{R}}

\newcommand{\C}{\mathbb{C}}

\newcommand{\Z}{\mathbb{Z}}

\newcommand{\Bc}{\mathcal{B}}

\newcommand{\Ic}{\mathcal{I}}

\newcommand{\Mc}{\mathcal{M}}

\newcommand{\Sc}{\mathcal{S}}

\newcommand{\Sf}{\mathfrak{S}}

\newcommand{\sgn}{\operatorname{sgn}}

\newcommand{\Span}{\operatorname{Span}}
\newcommand{\diag}{\operatorname{diag}}

\newcommand{\mat}[4]{
\left[\begin{array}{cc}
#1 & #2 \\ #3 & #4
\end{array}\right]}

\newcommand{\smat}[4]{
\left[\begin{smallmatrix}
#1 & #2 \\ #3 & #4
\end{smallmatrix}\right]}

\newcommand{\vmat}[2]{
\left[\begin{array}{cc}
#1 \\ #2
\end{array}\right]}

\newcommand{\abs}[1]{\lvert{#1}\rvert}

\newcommand{\ds}{\displaystyle} 

\newcommand{\itemspacing}{
\setlength\topsep{0.1in}
\setlength\itemsep{0.1in}}

\renewcommand{\vec}{\bm}
\renewcommand{\Sf}{\Bc}

\title{The Transfer of Polarised Radiation in Homogeneous Water Bodies}

\author{Andrew Corbett}

\newcommand{\Date}{11$^{\mathrm{th}}$ November 2019}
\date{\Date}

\begin{document}

\maketitle

\begin{abstract}
We give an analytic solution for the propagation of polarised radiation through a homogeneous water body. As corollaries we derive the vector bidirectional reflectance distribution function at the bottom of an infinitely deep water body and compute the asymptotic radiance distribution. These have applications to polarised radiative transfer simulation in the inhomogeneous setting. Of independent interest, we give a concise variant formulation of the azimuthal decoupling via a complex Fourier transform which is applicable to radiative transfer models in general.

\end{abstract}


\section{Introduction}
\label{sec:introduction}

The simulation of light propagation within a natural water body has long been a topic of great maritime and oceanographic importance. The problem is characterised by an equation of radiative transfer governing the propagation of light intensity
in such a water body \citep{chandrasekhar, preisendorfer, mobley-LAW} which may be solved via various numerical methods \citep{mobley-comp}.
More recently, the transfer of \textit{polarised} radiation has been approached with the same arsenal of solution techniques \citep{chami, tynes, mish-02, emde, zhai-09, mobley-VRTE}, generalising the preceding `scalar' theory by solving for the Stokes vector in the Vector Radiative Transfer Equation (VRTE).

In this work we give an analytic solution for the propagation of polarised radiation through an \textit{homogeneous} water body, one in which the inherent optical properties, such as the scattering behaviour, are independent of depth. This is an idealistic but fundamental prototype scenario in which we can derive an implementable (see \S \ref{sec:implementation}) analytic solution to the VRTE. The solution is indicative of the general case, giving structural insight to the propagation of polarised radiation. We also derive applications to the numerical simulation of polarised radiative transfer in a general inhomogeneous water body. Namely, we deduce the following:

\begin{enumerate}[(i)]

\item The Vector Bidirectional Reflectance Distribution Function (VBRDF) at a finite depth in an infinitely deep natural water body.

\item
The asymptotic radiance distribution for the (polarised) Stokes vector at depths approaching infinity.

\item
A compact construction for the Fourier decoupling of a discretely approximated azimuthal distribution, applicable more widely to various radiative transfer simulations in inhomogeneous water.
\end{enumerate}

Originally motivated by (i), the VBRDF of an infinitely deep water body (see \S \ref{sec:intro-bottom} and \S \ref{sec:infinite}) completes the set of boundary conditions for a full inhomogeneous vector radiative transfer model as set out by \cite{mobley-VRTE}.
We derive an efficient algorithm to execute its computation, superseding the existing approach of solving inhomogeneous equations to an arbitrarily large depth \citep[\S 1.2.6.2]{mobley-VRTE}.

However, our analysis may be applied more widely.
Simple water bodies may indeed be homogeneous over vast regions. The theoretical framework we provide below describes the solution to the VRTE in such water.
We provide details on implementing this solution numerically in \S \ref{sec:implementation}.
In future work we hope to extend this to apply to more general water bodies by viewing them as comprising of a collection of homogeneous areas (see \S \ref{sec:homog-patching}). Thus giving means by which to analyse data for polarised radiation.
We also note that the `vector' theory developed here contains the scalar theory as a special case: the radiance magnitude is given by the first component of the Stokes vector \eqref{eq:stokes-vector}.

The scattering of light in the scalar case is described by the `volume scattering function' \citep[Ch.\ 3]{mobley-LAW} for which there is well-studied and standardised experimental data \citep{petzold, fournier-forand}. For polarised light, such a data standard is simply not recorded in the literature. Our present work thus sits as a theoretical guide, in keeping with the principles of \cite{mobley-VRTE}. The collation of data regarding polarised scatting in water bodies is an important topic of research for the testing and calibration of numerical models such as our own.


We also give a new formulation of the discrete spectral VRTE which may be applied to discretely approximated radiative transfer methods in general; see \S \ref{sec:spectral-VRTE-main}.
The theoretical content of our derivation and that of \cite{mobley-VRTE} is conserved, however, by use of the complex valued discrete Fourier transform, our construction is functionally more concise.
Moreover, when applied to the vector case, as it is here, the mathematical formulation is more intuitive. Indeed, in the scalar case the volume scattering function depends only on the cosine of the difference of azimuthal coordinates, $\cos(\phi-\phi')$, say. Whereas for polarised radiation, the scattering phase matrix depends more generally on $\phi-\phi'$ itself (see \S \ref{sec:phase-matrix}) and thus its Fourier expansion does not depend solely on a cosine series.
Introducing a complex Fourier series with respect to the character $$x\mapsto e^{2\pi i x} = \cos(x) + i\sin(x),$$
we are able to neatly address the shift in scattering planes as well as homogenise anomalies in the overall derivation.
Many software packages operate fast parallel methods for complex multiplication and addition meaning that our method has an advantage in terms of computational implementation.

In the remainder of this introduction we shall discuss the key features of our results, stating them explicitly with proof to follow in the subsequent sections.

\subsection{An analytic solution for the Stokes vector}
\label{sec:intro-analytic-sol}

Polarised light is simulated by the \textit{incoherent\footnote{The Stokes vector is called incoherent when defined in terms of radiance values; that is, with respect to infinitesimal increments in direction of its defining beam.} Stokes vector}
$S$ consisting of four components, each evaluating the radiance
at a point, delivered in different polarisation states for a beam of radiation of fixed wavelength; see \S \ref{sec:stokes-vector}. If $S$ is considered at a positive depth $z$ in a plane-parallel natural water body then it satisfies a VRTE, a first order integro-differential equation, as given in \eqref{eq:VRTE-gen}.

Radiance itself is the measure of radiant power over infinitesimal increments in surface area, bandwidth and, characteristically, solid angle of incidence
with units of W m$^{-2}$ nm$^{-1}$ sr$^{-1}$. Its distribution over the spherical directions emanating from the given point at $z$ is parameterise by a polar and azimuthal coordinate, as detailed in \S \ref{sec:geometric-assuptions}. It is common practice to disassemble the azimuthal dependence of the VRTE via a discrete Fourier decoupling trick. We give a novel concise exposition of this in \S \ref{sec:deriving-spectral}. The upshot being that the analysis of $S$ is shifted to that of its complex Fourier coefficients $\hat{S}$ which satisfy independent VRTEs, decoupled in the azimuthal variable. We consider the polar variable discretely spaced over $2M$ subintervals, $M$ in each hemisphere, and suppose that the components of the VRTE are constant on these intervals. For example, by replacing them with their quad-averages \citep[\S 1.3.2]{mobley-VRTE}. The continuous VRTE in \eqref{eq:VRTE-gen} is then reduced to solving a system of $8M$ linear equations, $4$ components in each Stokes vector by $2M$ polar coordinates:
\begin{equation}
\label{eq:VRTE-intro}
\frac{d\hat{S}}{dz} = K(z) \hat{S}(z) + \hat{\sigma}(z)
\end{equation}
where $K(z)$ is an $8M\times 8M$ `transfer' matrix describing the elastic scattering behaviour of the water body and $\hat{\sigma}(z)$ a vector of $8M$ components corresponding to the (Fourier transform of the) external source term. The derivation of \eqref{eq:VRTE-intro} is given in \S \ref{sec:spectral-VRTE-main}.

\subsubsection{The general case}

Analytically, the control process in \eqref{eq:VRTE-intro} possesses a general solution of the form
\begin{equation}
\label{eq:intro-analytic-sol}
\hat{S}(z) = \Phi(z,z_0)\hat{S}(z_0) + \int_{z_0}^{z}\Phi(z,z')\hat{\sigma}(z')dz'.
\end{equation}
where $\Phi(z,z_0)$ denotes the state-transition matrix between depths $z$ and $z_0$.
In an inhomogeneous natural body of water with a two-point boundary condition, one at the surface and one at the bottom, $\Phi(z,z_0)$ is difficult to determine. One most efficiently solves \eqref{eq:VRTE-intro} via a numerical integration algorithm; for example, by first applying an `invariant imbedding identity' (see \eqref{eq:invariant-imbedding-reln}) to split the problem into two one-point (initial value) problems \citep{preisendorfer, mobley-LAW, mobley-VRTE} which may be integrated numerically.
Other approaches include the method of discrete ordinates \citep{stamnes-88, evans, emde} or considering successive orders of scattering \citep{chami, zhai-09, zhai-15, zhai-17}.

\subsubsection{Homogeneous water}

When the inherent optical properties of a water body are depth invariant the analytic solution of \eqref{eq:VRTE-intro} becomes tangible for a computer to evaluate.
This assumption amounts to that $K=K(z)$ and $\hat{\sigma}=\hat{\sigma}(z)$ are both constant functions $z$. On physical grounds \citep[p.\ 451]{mobley-LAW}, the eigenvalues of $K$ are distinct, meaning that $K$ is diagonalisable of the form $K = E D E^{-1}$ where $D$ is a diagonal matrix of eigenvalues and the rows of $E$ are the eigenvectors. In this case, the state-transition matrix is explicitly given by
\begin{equation}
\label{eq:intro-state-mat}
\Phi(z,z_0)=E\exp(D(z-z_0))E^{-1}.
\end{equation}
Then the analytic solution \eqref{eq:intro-analytic-sol} is solved up to the computation of $E$ and $D$. This may be solved by a numerical algorithm with computational time typically proportional to the cube of the order of $K$, in our case $8M$.
By hemispherical symmetry arguments, the matrix of eigenvectors is of the symmetric block form
\begin{equation}
\label{eq:intro-eigenvalue}
E=\begin{bmatrix}
E^{+}&E^{-}\\
E^{-}&E^{+}
\end{bmatrix}.
\end{equation}
This is derived in \S \ref{sec:hemi-sym} and may be compared to the scalar case \citep[\S 9.4]{mobley-LAW}. With this decomposition, we reduce the order $8M$ problem to one of order $4M$. Probing further, in \S \ref{sec:phase-sym} we deconstruct the polarised scattering matrices to, in certain cases, reduce the problem to one of order $3M$.

\subsection{Bottom reflectance in a deep natural water body}
\label{sec:intro-bottom}

Deep waters are time consuming to model numerically. However, it is common that after a certain depth there is a large homogeneous region of source-free water. We analytically compute the VBRDF of such a region in \S \ref{sec:infinite}. As a result, modelling an inhomogeneous needs only to take place down to a shallow depth above this region. Algorithmically, this is a great improvement over assuming a (diffusely scattering) Lambertian boundary condition at a very large depth.

The first and motivating corollary to the result described in \S \ref{sec:intro-analytic-sol} is the derivation of an explicit formula for the VBRDF in an infinitely deep water body. This is given by the reflectance $R(z_0,\infty)$ at some depth $z_0$, past which the water is assumed infinitely deep, source free and homogeneous (see \S \ref{sec:analytic-solution-reflectance} for a formal definition). Our result shows that polarised radiation behaves similarly in each component to the scalar case. In \S \ref{sec:infinite} we prove that
\begin{equation*}
R(z_0,\infty) = E^{+}(E^{-})^{-1}
\end{equation*}
so that the reflectance is known up to the determination of the eigenvalues of $K$.

\subsection{Asymptotic radiance distributions}

An interesting observation of radiative transfer in very deep waters is that the depth dependence decouples completely from the directional distribution. This manifests the exponential decay of radiance in all directions. In the vector case, this observation is also true for each component defining the various polarisation states. In \S \ref{sec:asymp} we give a mathematical description of this phenomenon,  showing in \eqref{eq:asymptotic distribution} that the asymptotic Stokes vector decays exponentially with respect to its smallest eigenvalue uniformly in each component.

\subsection{Implementation}
\label{sec:implementation}

Let us here give a short manual of instruction for the implementation of our solution in a homogeneous water body.

\subsubsection*{Step 1}
Identify the depth-independent phase scattering matrix $P$ in \eqref{eq:phase-matrix} and the source term $\sigma$ of the water body and take their Fourier transforms as in \eqref{sec:deriving-spectral}. Compute the local transmittance $\tau$ and reflection $\rho$ matrices as in \S \ref{sec:disc-notation}.

\subsubsection*{Step 2}
Compute the eigenvectors $\tilde{E}$ and eigenvalues $\tilde{D}$ of the auxiliary $4M$-dimensional system $(\tau-\rho)(\tau+\rho)$ as in \S \ref{sec:hemi-sym}. This is achieved via a numerical procedure. If the scattering matrix $P$ is of part-diagonal type (see \S \ref{sec:part-diagonal-type}) then this may be reduced to a $3M$-dimensional and an $M$-dimensional system.

\subsubsection*{Step 3}
Compute $D^{+}=\sqrt{\tilde{D}}$ and the differences $E^{+} - E^{-}= \tilde{E}$ and $E^{+}+E^{-}$ via \eqref{eq:sim-equations}. Then $D=\diag(D^{+},-D^{+})$ and $E^{+} = \frac{1}{2}(E^{+}+E^{-}+\tilde{E})$ and $E^{-} = \frac{1}{2}(E^{+}+E^{-}-\tilde{E})$ determine $E$ as in \eqref{eq:intro-eigenvalue}. With $E$ and $D$ evaluated, the explicit solution for the Fourier transform $\hat{S}$ is obtained via \eqref{eq:intro-analytic-sol} and \eqref{eq:intro-state-mat}. To invert the Fourier transform and arrive at the final solution apply \eqref{eq:fourier-inversion}.

\subsection{Homogeneous patching}
\label{sec:homog-patching}

Of course, not all natural waters may be assumed homogeneous. At the gain of computational speed, we speculate on simulating inhomogeneous waters roughly as a concatenation of homogeneous sections. Such homogeneous sections would then be `patched' together via the use of an invariant imbedding relation (for example, see \eqref{eq:invariant-imbedding-reln}).
The computational efficiency would then scale linearly with the number of sections considered. In a future work we plan to investigate such an implementation as a realistic alternative natural water body simulation.

\subsection{Common notation}
If $A$ is an $m\times n$ array we let $A^{T}$ denote its transpose, the $N\times m$ array in which the $i$-th row of $A$ is the $i$-th column of $A^{T}$. We extend the exponential map $\exp(x):=e^{x}$ more generally to $A$ via the Taylor expansion
\begin{equation*}
\exp(A):=\sum_{k\geq 0} \frac{A^k}{k!}.
\end{equation*}
In particular, we introduce a special notation to describe the complex unit circle, $e(x):=\exp(2\pi i x) = \cos(x) + i \sin(x).$
We use the Kronecker delta indicator
\begin{equation*}
\delta_{x,y} := \begin{cases}	1 & \text{if } x=y\\ 0 & \text{if } x\neq y\end{cases}
\end{equation*}
to write the $n\times n$ identity matrix as $\bm{1}_{n}:=[\delta_{i,j}]_{1\leq i,j\leq n}$ and a general block diagonal matrix as $\diag(A_1,\ldots,A_n):=[\delta_{i,j}A_i]_{1\leq i,j\leq n}$ for coefficient matrices $\{A_i\}_{1\leq i\leq n}$.

\section{The transfer of polarised radiation in a natural water body}
\label{sec:VRTE}

We begin with some background on the Stokes vector and its governing VRTE, the definitions of which depend sensitively on their geometric structure.

\subsection{Geometric assumptions}
\label{sec:geometric-assuptions}

\subsubsection{Depth scales and Cartesian geometry}
\label{sec:cartesian-basis}

Fix the plane $\Sc$ defined by the mean surface of the water body and on $\Sc$ choose the depth to be $z=0$. From here, choose the $\vec{z}$-axis to point downwards so that $z>0$ increases positively with depth into the water.
Pick linearly independent unit vectors $\vec{x}$ and $\vec{y}$ to span $\Sc$ as follows: by convention, take $\vec{x}$ to point downwind and choose $\vec{y}$ to satisfy the right-hand system $\vec{z}=\vec{x}\times\vec{y}$. 

We assume any water body to be `plane-parallel' which is to say that light propagation is invariant within any plane parallel to $\Sc$ at a fixed positive depth. We model radiative transfer by the propagation through an infinitesimally thin column of water in the $\vec{z}$-direction.

\subsubsection{Spherical coordinates}
\label{sec:spherical-coords}

The possible directions into which radiation may be emitted are contained in the unit sphere $\Sf$. We parameterise $\Sf$ by the coordinates
\begin{equation*}
\Sf=\{ (\mu,\phi)\, \vert\,  \mu\in [-1,1],\, \phi\in[0,2\pi)\}
\end{equation*}
where $\mu$ is the \textit{polar} coordinate, the projection of the direction onto the $\vec{z}$-axis, and $\phi$ is the \textit{azimuth} which rotates clockwise in a plane parallel to $\Sc$ when looking in the positive $\vec{z}$-direction and is fixed such that $\phi = 0$ points along $\vec{x}$, downwind.
Explicitly, an arbitrary direction in $\Sf$ centred at $(0,0,0)$ is given by the unit vector
\begin{equation}
\label{eq:rhat}
\vec{r} =\vec{r}(\mu,\phi) :=  \sin(\arccos(\mu))\cos(\phi)\vec{x}+\sin(\arccos \mu)\sin(\phi)\vec{y}+\mu\vec{z}.
\end{equation}
We distinguish between the `downward' hemisphere $\Sf^{+} := \{(\mu,\phi)\in \Sf \,\vert \,\mu\geq 0\}$ and the `upward' hemisphere $\Sf^{-} := \{ (\mu,\phi)\in \Sf \,\vert \,\mu\leq 0\}$ so that
$\Sf=\Sf^{+}\cup \Sf^{-}.$

In addition to $\{\vec{x},\vec{y},\vec{z}\}$, we consider a second orthonormal basis of $\R^{3}$ for each direction $(\mu,\phi)\in\Sf$. Let $\vec{\varphi}=\vec{\varphi}(\mu,\phi):= -\sin(\phi) \vec{x} + \cos(\phi) \vec{y}$, a unit vector parallel to $\Sc$, and choose the unit vector $\vec{\nu}=\vec{\nu}(\mu,\phi)$ satisfying $\vec{r} = \vec{\nu} \times \vec{\varphi}$ with $\vec{r}$ defined as in \eqref{eq:rhat}; explicitly, $\vec{\nu}(\mu,\phi) := \mu\cos(\phi) \vec{x} + \mu\sin(\phi) \vec{y} - \sin(\arccos\mu)\vec{z}$.
The basis $\{\vec{r},\vec{\nu},\vec{\varphi}\}$ is then better orientated to determine the polarisation states of a beam propagating along $\vec{r}$.
We call the plane $\Mc(\mu,\phi) := \Span_{\R}(\vec{r},\vec{\nu})$ the \textit{meridian plane} with respect to the direction $(\mu,\phi)\in\Sf$.

\subsection{The Stokes Vector}
\label{sec:stokes-vector}

Consider a beam of radiation with point-wise propagation direction $(\mu,\phi)\in\Sf$. Its associated electric field solution $E$ oscillates orthogonally to its propagation direction in the plane spanned by $\vec{\nu}$ and $\vec{\varphi}$. So at a fixed time and position we have $E= E_{\nu}\vec{\nu} + E_{\varphi}\vec{\varphi}$ with $E_{\nu},E_{\varphi}\in\C$. By convention, these two orthogonal components determine two linear polarisation states; these are termed the `vertical' ($E_{\varphi}=0$) and `horizontal' ($E_{\nu}=0$) polarisations states\footnote{Only in a scattering event do we refer to $\vec{s}$-polarised and $\vec{p}$-polarised light with respect to the scattering plane (of incidence). Such a coordinate transformation is described in \S \ref{sec:phase-matrix}.}
\citep[p.\ 16]{mish-02}.
We define the four-component (coherent) \textit{Stokes vector}
\begin{equation}
\label{eq:stokes-vector}
S:=
\frac{1}{2\eta}
\left[
{\arraycolsep=1.4pt\def\arraystretch{1.1}
\begin{array}{rcl}
E_{\nu}\bar{E}_{\nu}&+&E_{\varphi}\bar{E}_{\varphi}\\
E_{\nu}\bar{E}_{\nu}&-&E_{\varphi}\bar{E}_{\varphi}\\
-(E_{\varphi}\bar{E}_{\nu}&+&E_{\nu}\bar{E}_{\varphi})\\
i(E_{\varphi}\bar{E}_{\nu}&-&E_{\varphi}\bar{E}_{\nu})
\end{array}
}
\right],
\end{equation}
where $\eta$ is the characteristic impedance of the medium, 
to describe the intensity (or rather irradiance) delivered in each polarisation state. However, in this work we shall exclusively consider the \textit{incoherent} Stokes vector by defining the direction of $E$ only up to incremental amounts via the differential $d\mu\times d\phi$ on $\Sf$. The components of incoherent Stokes vector are then values of radiance. In particular, the first component is precisely the scalar (magnitude) radiance whereas the others measure the linearity and helicity of the radiance.
For example, linearly polarised light propagates in proportion to
$[1,\pm 1,0,0]^{T}$.

\subsection{The continuous depth-dependant VRTE}
\label{sec:VRTE-gen}

Within a plane-parallel, temporally invariant water body, the propagation of monochromatic polarised radiation is described by the incoherent Stokes vector $S=S(z;\mu,\phi)$, as in \S \ref{sec:stokes-vector}, for depths $z>0$ as a distribution over $(\mu,\phi)\in\Sf$. The water body in which it propagates forces upon $S$ the following VRTE \citep[(1.5)]{mobley-VRTE}:
\begin{multline}
\label{eq:VRTE-gen}
\mu \frac{dS}{dz}(z;\mu,\phi) = 
- c(z) S(z;\mu,\phi)\\
+ \int_{0}^{2\pi}\int_{-1}^{1} P(z; (\mu',\phi'){\rightarrow}(\mu,\phi)) S(z;\mu',\phi') d\mu'd\phi' + \sigma(z;\mu,\phi)
\end{multline}
where $P$ denotes the $4\times 4$ elastic scattering phase matrix, describing the scattering of polarised light incident at $z$ in direction $(\mu',\phi')$ scattering into resultant beams of direction $(\mu,\phi)$; the term $c$ denotes the beam attenuation coefficient; and $\sigma$ denotes an arbitrary internal source term \citep[\S 1.2.3]{mobley-VRTE}. These terms are the Inherent Optical Properties (IOPs) of the system: $c$ and $P$ completely describe the absorbing and elastic scattering behaviour of the water body, respectively, whereas $\sigma$ accounts for inelastic scattering processes from other wavelengths (Ramen scattering or fluorescence) or bioluminescence.

\subsection{Coordinates for polarised scattering processes}
\label{sec:scattering-coords}

Consider a beam, incident at a point, in direction $(\mu',\phi')$ whose resultant scattered light is distributed over directions $(\mu,\phi)$. Denote the incident and resultant directions by $\vec{r}':=\vec{r}(\mu',\phi')$ and $\vec{r}:=\vec{r}(\mu,\phi)$ respectively, as in \eqref{eq:rhat}. Moreover, following \S \ref{sec:spherical-coords}, write $\vec{\nu}':=\vec{\nu}(\mu',\phi')$, $\vec{\varphi}':=\vec{\varphi}(\mu',\phi')$, $\vec{\nu}:=\vec{\nu}(\mu,\phi)$, and $\vec{\varphi}:=\vec{\varphi}(\mu,\phi)$.

\subsubsection{The scattering plane}

We call the angle $\psi\in[0,\pi]$ between $\vec{r}'$ and $\vec{r}$ \textit{scattering angle}. 
By definition, $\psi$ satisfies
\begin{equation}
\label{eq:scattering-angle}
\cos(\psi) = \vec{r}(\mu',\phi')\cdot\vec{r}(\mu,\phi) = \mu'\mu + \sqrt{(1-\mu'^2)(1-\mu^2)} \cos(\phi-\phi').
\end{equation}
This angle is measured in the plane spanned by $\vec{r}'$ and $\vec{r}$, the \textit{scattering plane}. It is with respect to this plane that we define the operations of the scattering matrix $P$.

\subsubsection{Senkrecht and parallel axes}

Make a choice of a pair of axes $(\vec{s}',\vec{p}')$ and $(\vec{s},\vec{p})$ with components which are orthogonal (senkrecht), $\vec{s}',\vec{s}$ and parallel $\vec{p}',\vec{p}$ to the scattering plane such that $\vec{s}'\times \vec{p}'=\vec{r}'$ and $\vec{s}\times \vec{p}=\vec{r}$. This choice is then uniquely defined by fixing the sign of $\vec{p}'$ and $\vec{p}$ so that they are the first parallel vectors found after rotating $\vec{\nu}'$ and $\vec{\nu}$, respectively, in an anti-clockwise direction when looking into the beam by convention.\footnote{There is no conflict of terminology for the vertical and horizontal unit vectors between ours and that of \cite{mobley-VRTE}. In particular, we rotate the azimuth $\phi$ clockwise, not anti-clockwise $-\phi$. Under this mapping, $\vec{\nu}(\mu,\phi)=-\vec{\nu}(\mu,-\phi)$ is equal to the horizontal unit vector on \citep[p.\ 8]{mobley-VRTE}. We refer to \citep[p.\ 14]{mobley-VRTE} for a more elaborate discussion on conventions amongst various authors.}

\subsubsection{Affine transformations}

A scattering event described by $P$ concerning an incident direction $(\mu',\phi')$ and a resultant distribution $(\mu,\phi)$ is defined with respect to the axes $(\vec{s}',\vec{p}')$ and $(\vec{s},\vec{p})$ as given above. However, the polarisation states encoded in the Stokes vector are defined according to the horizontal and vertical coordinates with respect to their respective meridian planes (see \S \ref{sec:stokes-vector}).
To this end we consider an affine transformation to rotate these axes from the incident system, to the scattering plane, and then into the resultant system after scattering has occurred.

Let $\alpha',\alpha \in [0,2\pi)$ denote the angles between $(\vec{\nu}',\vec{\varphi}')$ and $(\vec{p}',\vec{s}')$; and between $(\vec{\nu},\vec{\varphi})$ and $(\vec{p},\vec{s})$, respectively. These may be deduced, at least up to sign, from the identities $\cos(\alpha') = \vec{\nu}'\cdot \vec{p}'$ and $\cos(\alpha) = \vec{\nu}\cdot \vec{p}$; they may be explicated further, as in \citep[(1.6)-(1.13)]{mobley-VRTE}, but we do not require such an expression in this work.
To execute the affine transformation, or $\gamma\in\R$ define the rotation matrix
\begin{equation}
\label{eq:rotation-matrix}
R(\gamma) =
\left[
{\arraycolsep=1.5pt\def\arraystretch{1.1}
\begin{array}{rrrr}
1	&					&					&	\\
	& \cos(2\gamma)		& -\sin(2\gamma)		&	\\
	& \sin(2\gamma)		& \cos(2\gamma)		&	\\
	&					&					& 1
\end{array}
}
\right].
\end{equation}
We record the following useful identity used later for computational efficiency.

\begin{lem}
\label{lem:coords-flip}
With respect to the directions $(\mu',\phi')$ and $(\mu,\phi)$ as above, consider the mapping given by $(\phi,\phi')\mapsto (\pi - \phi,\pi - \phi')$ and either (i) $(\mu,\mu')\mapsto (-\mu,\mu')$ or (ii) $(\mu,\mu')\mapsto (\mu,-\mu')$. Then the corresponding rotation angles respectively satisfy
\begin{enumerate}[(i)]
\item $(\alpha,\alpha')\mapsto(\alpha,\pi - \alpha')$ and $(R(\alpha),R(\alpha'))\mapsto(R(\alpha),R(-\alpha'))$;
\item $(\alpha,\alpha')\mapsto(\pi - \alpha,\alpha')$ and $(R(\alpha),R(\alpha'))\mapsto(R(-\alpha),R(\alpha'))$;
\end{enumerate}
and in both cases $\cos(\psi)\mapsto\cos(\pi-\psi)$.
\end{lem}

\begin{proof}
This follows by direct computation. The trigonometric functions in \eqref{eq:rotation-matrix} may be explicitly expressed in terms of $(\mu',\phi')$ and $(\mu,\phi)$ by applying the spherical trigonometric identities. This is recorded in \citep[\S 1.2.3.1]{mobley-VRTE}. To these expressions one applies the identities $\cos(\pi-x)=-\cos(x)$ and $\sin(\pi-x)=\sin(x)$.
\end{proof}

\subsection{The scattering phase matrix}
\label{sec:phase-matrix}

\subsubsection{Scattering due to mirror symmetric particles}

We impose the following assumption on the scattering media that we consider.
\begin{ass}
\label{ass:mirror-sym}
Particles within the water body are randomly oriented and mirror-symmetric.
\end{ass}
The upshot of Assumption \ref{ass:mirror-sym} is that at a fixed depth $z>0$ the scattering matrix, $P(z; (\mu',\phi'){\rightarrow}(\mu,\phi))$, depends only on the following:
\begin{itemize}
\itemspacing
\item The scattering angle $\psi$ as given in \eqref{eq:scattering-angle}.
\item The affine transformations $R(\alpha)$ and $R(\alpha')$ as in \eqref{eq:rotation-matrix}.
\end{itemize}
Explicitly, as computed by \cite{mish-02} under Assumption \ref{ass:mirror-sym}, we have the decomposition
\begin{equation}
\label{eq:phase-matrix}
P(z; (\mu',\phi'){\rightarrow}(\mu,\phi)) = R(\alpha) M(z;\psi) R(\alpha'),
\end{equation}
in which we implicitly define the \textit{scattering matrix} $M(z;\psi)$; this has a special block diagonal structure of the form $M=\diag(M_1,M_2)$ for where $M_{1}$ and $M_{2}$ are $2\times 2$ real-valued matrices satisfying the relations
$M_{1}=M_{1}^{T}
\text{ and }
M_{2}=\smat{-1}{}{}{1}M_{2}^{T}\smat{-1}{}{}{1}.$

\subsubsection{Scattering matrices of part-diagonal type}
\label{sec:part-diagonal-type}

An important class of scattering matrices are classified by the additional constraint that either $M_1$ or $M_2$ are themselves a diagonal matrices. We say such a scattering matrix is of \textit{Part-diagonal type}. A prototype example is given by the \textit{Rayleigh scattering} process; Rayleigh scattering describes the elastic scattering of light from molecular sized particles whose diameter is less than the wavelength of the incident light by an order of magnitude.  The associated scattering matrix is of the form
\begin{equation*}
M(\psi)=M_{\text{Ray}}(\psi) := 
b_{\text{Ray}}(z)
\begin{bmatrix}
1+\cos(\psi)^{2}	&	\sin(\psi)^{2}		&				&\\
-\sin(\psi)^{2}		&	1+\cos(\psi)^{2}	&				&\\
					&						&	2\cos(\psi)	&\\
					&						&				&	2\cos(\psi)
\end{bmatrix}
\end{equation*}
where $b_{\text{Ray}}(z)$ is a depth dependant constant given by the average of the top-left entry of $M_{\text{Ray}}$ over all directions.
We shall find that scattering matrices of part-diagonal type may be computed with greater computational efficiency.

\subsubsection{Scattering matrix symmetries}
\label{sec:phase-sym}

The dependence of the scattering phase matrix $P(z;(\mu',\phi'){\rightarrow}(\mu,\phi))$ on its azimuthal variables $\phi$ and $\phi'$ is determined by the quantity $\phi-\phi'$. Moreover, the polar coordinates depend only on their relative hemisphere position; that is, on whether $\mu\mu'>0$ (same hemisphere) or $\mu\mu'<0$ (different hemispheres). These symmetries are summarised as follows:
\begin{equation}
\label{eq:phase-sym}
\begin{array}{rcl}\vspace{0.1in}
P(z;(\mu',\phi'){\rightarrow}(\mu,\phi))&=&
P(z;(\mu',0){\rightarrow}(\mu,\phi-\phi'))\\
&=& P(z;(\sgn(\mu)\mu',\phi'){\rightarrow}(\sgn(\mu')\mu,\phi))
\end{array}
\end{equation}
where $\sgn(\mu):=\mu/\abs{\mu}$. It is however important to note that the symmetry $(\phi,\phi')\mapsto(\phi',\phi)$ is not preserved by $P$, in contrast to the scalar case.

\section{The discrete spectral VRTE}
\label{sec:spectral-VRTE-main}


A key exploit in solving the VRTE given in \eqref{eq:VRTE-gen} is to decouple the azimuthal integral over $\phi'\in[0,2\pi)$ from the fixed direction $\phi$. This is made possible via the symmetries of the scattering matrix $P$ (see \S \ref{sec:phase-sym}). This may be realigned to remove the double azimuthal dependence via a phase shift in its Fourier transform. This constitutes the `spectral' aspect of this section.

The `discrete' aspect is due to the assumption that all $\Sf$-distributions in \eqref{eq:VRTE-gen} behave as `step functions' over an agreed partition (or `quad averaging') of $\Sf$.
Specifically, this assumption is applied to $S$, $\sigma$ and both variables of $P$. 
In this section, it is then a \textit{finite} Fourier transform that we apply to the azimuthal arguments and thus derive a VRTE for each of their respective discrete spectral transforms. The system of transformed equations is then decoupled and each may be solved independently. This standard practice is applicable to all methods for solving radiative transfer equations. Our exposition here is somewhat different to the literature, and we hope it shall provide a quicker route through the derivation.

\subsection{Discrete approximation of the azimuthal coordinate}

Fix a positive integer $N$. We firstly approximate the circles defined by $\phi\mapsto\vec{r}(\mu,\phi)$ for each $-1\leq \mu\leq 1$ by an $N$-gon, reducing analysis on the torus $\R/2\pi \Z$ to that on the finite additive cyclic group $\Z/N\Z$.
Define $N$ equally spaced azimuthal coordinates by
\begin{equation*}
\phi_{v}:= \frac{2\pi v}{N}
\end{equation*}
for $v=0,\ldots,N-1$. These determine the partition
$[0,2\pi) = \bigsqcup_{v=0}^{N-1}[\phi_{v},\,\phi_{v}+\tfrac{2\pi}{N}).$ Identifying the set $[0,2\pi)$ with $\R/2\pi\Z$, the map $v\mapsto \phi_{v}$ gives an embedding $\Z/N\Z \hookrightarrow \R/2\pi\Z$. We proceed by only considering functions which pull back uniquely to this finite group.

\begin{ass}[Azimuthal step-function approximation]
\label{ass:step-az}
The three mappings given by $\phi\mapsto S(z;\mu,\phi)$, $\phi\mapsto \sigma(z;\mu,\phi)$ and $(\phi',\phi)\mapsto P(z;(\mu',\phi'){\rightarrow}(\mu,\phi))$ are constant on the intervals $ [\phi_{v},\,\phi_{v}+\tfrac{2\pi}{N})$ for $0\leq v\leq N-1$.
\end{ass}

\subsection{The finite Fourier transform}
\label{sec:finite-fourier}

A periodic function $f\colon\R/2\pi\Z\rightarrow\C$ is naturally restricted to a function $\Z/N\Z\rightarrow\C$ via $v\mapsto f(\phi_v)$. For each integer $0\leq\ell\leq N-1$, we thus define the \textit{finite Fourier transform} of $f$ by
\begin{equation}
\label{eq:fourier-transform}
\hat{f}(\ell):=\sum_{v=0}^{N-1}f(\phi_v)e(-\ell v/N),
\end{equation}
recalling that $e(x)=e^{2\pi i x}$ for $x\in\R$.
For $0\leq v \leq N-1$, the orthogonality relation 
\begin{equation}
\label{eq:orthog-basic}
\frac{1}{N}\sum_{\ell=0}^{N-1}e((\ell-v)/N)=
\delta_{v,0}
\end{equation}
implies the inversion formula
\begin{equation}
\label{eq:fourier-inversion}
f(\phi_v)=\frac{1}{N}\sum_{\ell=0}^{N-1} \hat{f}(\ell)e(\ell v/N).
\end{equation}
If $f(x)\in\R$ for each $x\in\R$, we use the Euler's identity, $e(x)=\cos(x) + i \sin(x)$, to evaluate the real part of \eqref{eq:fourier-transform} so that the expansion \eqref{eq:fourier-inversion} is expressed in terms of trigonometric functions.

\subsection{Deriving the spectral VRTE}
\label{sec:deriving-spectral}
We now change variables in \eqref{eq:VRTE-gen} to instead consider the transforms
\begin{equation}
\label{eq:fourier-transform-S}
\hat{S}(z;\mu;\ell):=\sum_{v=0}^{N-1}S(z,\mu,\phi_{v})e(-\ell v/N);\,\,
\hat{\sigma}(z;\mu;\ell):=\sum_{v=0}^{N-1}\sigma(z,\mu,\phi_{v})e(-\ell v/N);
\end{equation}
and
\begin{equation*}
\hat{P}(z;\mu,\mu';\ell) := \sum_{v=0}^{N-1}P(z;(\mu',0){\rightarrow}(\mu,\phi_{v}))e(-\ell v/N).
\end{equation*}
Note that $\hat{P}$ is the transform of $P$ in only one azimuthal variable in light of \eqref{eq:phase-sym}.
Substituting these transforms into \eqref{eq:VRTE-gen} via the inversion formula \eqref{eq:fourier-inversion}, applying Assumption \ref{ass:step-az} and \eqref{eq:phase-sym} to evaluate the $\phi'$-integral as a summation, we obtain
 \begin{multline}
 \label{eq:spectral-sub-1}
\mu\frac{d}{dz} \sum_{\ell=0}^{N-1} \hat{S}(z;\mu;\ell)e(\ell v/N) =  -c(z)\sum_{\ell=0}^{N-1} \hat{S}(z;\mu;\ell)e(\ell v/N)\\
 +\, \frac{2\pi}{N^2}
\sum_{v'=0}^{N-1}
 \,\int_{-1}^{1}
\,\sum_{\ell=0}^{N-1}\hat{P}(z;\mu,\mu';\ell)e(\ell(v - v')/N)
\sum_{\ell'=0}^{N-1} \hat{S}(z;\mu';\ell')e(\ell' v'/N)d\mu' \\
+\,\sum_{\ell=0}^{N-1} \hat{\sigma}(z;\mu;\ell)e(\ell v/N).
\end{multline}
Reordering the $(v',\ell,\ell')$ summations as $(\ell,\ell',v')$ and evaluating the orthogonal $v'$-sum per \eqref{eq:orthog-basic}, Equation \eqref{eq:spectral-sub-1} cleans up to
\begin{multline}
\label{eq:spectral-sub-2}
\mu\sum_{\ell=0}^{N-1} e(\ell v/N)\,\frac{d}{dz}\hat{S}(z;\mu;\ell) = -c(z)\sum_{\ell=0}^{N-1} e(\ell v/N)\,
\hat{S}(z;\mu;\ell) \\
+\,\sum_{\ell=0}^{N-1}e(\ell v/N)
\left( \frac{2\pi}{N}
\int_{-1}^{1}
\hat{P}(z;\mu,\mu';\ell) \hat{S}(z;\mu';\ell)d\mu' +  \hat{\sigma}(z;\mu;\ell)\right).
\end{multline}

\begin{lem}[Equating $\ell$-coefficients]
\label{lem:equating-coeffs}
For each $v=0,\ldots, 1-N$ suppose that $c_{\ell}\in\C$ satisfies $\sum_{\ell=0}^{N-1}e(\ell v/N) c_{\ell}  = 0$. Then $c_{\ell}=0$ for all $\ell=0,\ldots,N-1$.
\end{lem}

\begin{proof}
Consider the $N\times N$ matrix $A:=(e((i-1)(j-1)/N))_{1\leq i,j\leq N}$ and  $N\times 1$ column vector $C:=(c_{j-1})_{1\leq j\leq N}$. By the hypothesis we have that $AC = 0_{N\times 1}.$ Moreover, the Vandermonde matrix $A$ is invertible, with inverse $A^{-1}$, as is seen by evaluating the determinant
\begin{equation*}
\det(A)=\prod_{1\leq i< j\leq N} (e(j-1)-e(i-1))\neq 0.
\end{equation*}
The lemma now follows from the equality $C=A^{-1} 0_{N\times 1}=0_{N\times 1}.$
\end{proof}


By Lemma \ref{lem:equating-coeffs} this assertion the $\ell$-summands in \eqref{eq:spectral-sub-2} now decouple as follows.

\begin{prop}
\label{prop:VRTE-spec}
For $\mu\in[-1,1]$ and $\zeta=\zeta(z)$ for some $z\in(m,\infty)$ we have a decoupled spectral VRTE for each $0\leq \ell\leq N-1$ given by
\begin{equation*}
\mu\frac{d}{dz}\hat{S}(z;\mu;\ell) = -c(z)
\hat{S}(z;\mu;\ell)
+ \frac{2\pi}{N}
\int_{-1}^{1}
\hat{P}(z;\mu,\mu';\ell) \hat{S}(z;\mu';\ell)d\mu' + \hat{\sigma}(z;\mu;\ell).
\end{equation*}
\end{prop}

\begin{rem}
The advantage of using the spectral VRTE in Proposition \eqref{prop:VRTE-spec} over the original VRTE in \eqref{eq:VRTE-gen} is that the $\phi'$-integral in \eqref{eq:VRTE-gen} is eliminated, or rather `decoupled' from the dependence on $\phi$. The upshot being that the equations indexed by $0\leq \ell\leq N-1$ may be solved independently of one another. This is computationally efficient.
\end{rem}

\subsection{Discrete approximation of the polar coordinate}

Fix a positive integer $M$. We symmetrically partition the polar $\mu$-axis $[-1,1]$ into $2M$ subintervals over the hemispheres $\Sf^{\pm}$. The partition within each hemisphere is given arbitrarily since we do not make use of underlying structure or periodicity here. A specific partition may be chosen by the reader to be physically meaningful.

To this end, pick $M$ subintervals $\Ic_{i}\subset (0,1)$ such that $(0,1) = \bigsqcup_{i=1}^{M}\Ic_{i}.$ Also define their reflections about $0$ by $\Ic_{-i}:=-\Ic_{i}$. These partition the polar axis within the hemispheres $\Sf^{\pm}$, respectively. In each interval, fix an arbitrary choice of base point $\mu_{i}\in\Ic_{i}$ for $i=\pm1,\ldots,\pm M$. By definition we have $\mu_{-i}=-\mu_{i}$. We denote the length of an interval $\Ic_{i}$ by $\Delta\mu_{i}:=\int_{\Ic_{i}}d\mu.$

\begin{ass}[Polar step-function approximation]
\label{ass:step-po}
The mappings given by $\mu\mapsto S(z;\mu,\phi)$, $\mu\mapsto \sigma(z;\mu,\phi)$ and $(\mu,\mu')\mapsto P(z;(\mu',\phi'){\rightarrow}(\mu,\phi))$ and are constant on the intervals $ \Ic_{ i}$ for each $i = \pm 1,\ldots, \pm M$.
\end{ass}

\subsection{The discrete spectral VRTE as a linear system}
\label{sec:disc-notation}
To recognise the symmetry in our choice of intervals $\Ic_{i}$ and the scattering matrix symmetries in \eqref{eq:phase-sym}, we apply $\pm$ to our notation and henceforth consider $1\leq i,j \leq M$ by introducing
\begin{equation}
\label{eq:pm-notation}
\begin{array}{rc}
&\hat{S}^{\pm}(z;\mu_{i};\ell) := \hat{S}(z;\mu_{\pm i};\ell);\quad
\hat{\sigma}^{\pm}(z;\mu;\ell):= \hat{\sigma}(z;\mu_{\pm};\ell);\\
\text{and}&\\
&\hat{P}^{\pm}(z;\mu_{i},\mu_{j};\ell) :=
\hat{P}(z;\mu_{i},\mu_{\pm j};\ell)
= \hat{P}(z;\mu_{-i},\mu_{\mp j};\ell) .
\end{array}
\end{equation}
Then, under Assumption \ref{ass:step-po}, Proposition \ref{prop:VRTE-spec} implies that
\begin{multline}
\label{eq:VRTE-hemi}
\vspace{0.15in}
\pm\mu_{i}\frac{d}{dz}\hat{S}^{\pm}(z;\mu_{i};\ell) = 
-c(z)\hat{S}^{\pm}(z;\mu_{i};\ell)
\\
+\,\frac{2\pi}{N}\sum_{j=1}^{M} \Delta\mu_{j}
\left(\hat{P}_{\infty}^{\pm}(\mu_{i},\mu_{j};\ell) \hat{S}^{+}(z;\mu_{j};\ell)+\hat{P}_{\infty}^{\mp}(\mu_{i},\mu_{j};\ell) \hat{S}^{-}(z;\mu_{j};\ell)\right)
+ \hat{\sigma}^{\pm}(z;\mu_{i};\ell)
.
\end{multline}
By \eqref{eq:VRTE-hemi} we refer to $2M$ equations, each in the $4$ dimensions of the Stokes vector. We introduce a compact vector notation to describe this $8M$-dimensional linear system. For each $\ell=0,\ldots,N-1,$ define the $4M\times 1$ column vectors consisting of the $M$ transformed Stokes vectors and source terms, each having elements of dimension $4\times 1$:
\begin{equation*}
\hat{S}^{\pm}(z;\ell) := [\hat{S}^{\pm}(z;\mu_{i};\ell)]_{1\leq i \leq M},\,\,
\hat{\sigma}^{\pm}(z;\ell) := [\hat{\sigma}^{\pm}(z;\mu_{i};\ell)]_{1\leq i \leq M}.
\end{equation*}
Also define the $4M \times 4M$ `local transfer matrices' by the following $M\times M$ arrays each with $4\times 4$ entries:
\textit{the local transmittance matrices}
\begin{equation*}
\tau(z;\ell):=\left[\frac{\frac{2\pi}{N}\Delta\mu_{j}\hat{P}^{+}(z;\mu_{i},\mu_{j};\ell)-c(z)\delta_{i,j}\bm{1}_{4}}{\mu_{i}} \right]_{1\leq i,j\leq M}
\end{equation*}
and the \textit{local reflectance matrices}
\begin{equation*}
\rho(z;\ell):=\left[\frac{\frac{2\pi}{N}\Delta\mu_{j}\hat{P}^{-}(z;\mu_{i},\mu_{j};\ell)}{\mu_{i}} \right]_{1\leq i,j\leq M}.
\end{equation*}
These are so called due to their contribution preserving the radiance distribution in a given hemisphere: radiation scattered by the matrix $\hat{P}^{+}$ remains in the same hemisphere, adjusting its direction within that hemisphere according to $\tau$; on the other hand, radiation scattered by $\hat{P}^{-}$ makes a U-turn and is reflected into the opposite hemisphere via $\rho$.
Altogether, Equation \eqref{eq:VRTE-hemi} may now be succinctly written as
\begin{equation}
\label{eq:VRTE-matrix-pm}
\pm\frac{d}{dz}\hat{S}^{\pm}(z;\ell) = \tau(z;\ell) \hat{S}^{\pm}(z;\ell) + \rho(z;\ell) \hat{S}^{\mp}(z;\ell)  + \hat{\sigma}^{\pm}(z;\ell).
\end{equation}

Condensing notation further, we henceforth suppress the $\ell$-dependence and define the (complete) $2\times 1$ columns of $4M\times 1$ vectors
\begin{equation*}
\hat{S}(z) := \vmat{\hat{S}^{+}(z;\ell)}{\hat{S}^{-}(z;\ell)},\,\, \hat{\sigma}(z) := \vmat{
\hat{\sigma}^{+}(z;\ell)}{
\hat{\sigma}^{-}(z;\ell)}
\end{equation*}
alongside the \textit{local transfer matrix}, the $2\times 2$ array of $4M\times 4M$ matrices
\begin{equation}
\label{eq:K-def}
K(z):=
\left[
{\arraycolsep=1.5pt\def\arraystretch{1.1}
\begin{array}{rr}
\tau(z;\ell)&\rho(z;\ell)\\
-\rho(z;\ell)&-\tau(z;\ell)
\end{array}
}\right].
\end{equation}
We finally pose the homogeneous, discrete spectral VRTE as an $8M=4\times 2M$-dimensional system of (in general) non-linear equations
\begin{equation}
\label{eq:VRTE-8M}
\hat{S}'(z):=\frac{d}{dz}\hat{S}(z) = K(z)\hat{S}(z) + \hat{\sigma}(z).
\end{equation}
We shall later specialise to solving this system in homogeneous waters. This amounts to the following.
\begin{ass}[Homogeneous water]
\label{ass:homog}
The scattering phase matrix $P$, attenuation coefficient $c$, and source term $\sigma$, and consequently the local transfer matrix $K$ and transform $\hat{\sigma}$ in \eqref{eq:VRTE-8M}, are constant as functions of depth $z$.
\end{ass}

\section{A analytic solution via state transition}

We can describe the (discrete) process of radiative transfer through a water body as a multivariable control process of order $8M$. In this section we deduce an analytic solution to \eqref{eq:VRTE-8M} in both the non-homogeneous and homogeneous cases. Behind the scenes in this section, we have fixed $0\leq \ell\leq N-1$, which has been subverted from the notation in \eqref{eq:VRTE-8M}.

\subsection{The non-homogeneous analytic solution}
\label{sec:fundamental-sol}

An analytic solution for the `state vector' $\hat{S}(z)$ with respect to an initial value $\hat{S}(z_0)$ at $z_0>0$ is derived by introducing a \textit{state-transition matrix} $\Phi(z,z_0):= F(z)F(z_0)^{-1}$ where the $8M\times8M$ matrix $F(z)$ is a `fundamental solution' to the associated homogeneous problem
$\frac{dF}{dz}(z)= K(z)F(z).$
This non-explicit definition directly implies the following properties:
$\Phi(z_0,z_0) = 1$; $\Phi(z,z_0)=\Phi(z,z')\Phi(z',z_0)$ for $z'>0$; $\Phi(z,z_0)=\Phi(z_0,z)^{-1}$ and
\begin{equation*}
\frac{d}{dz}\Phi(z,z_0) = K(z)\Phi(z,z_0).
\end{equation*}
Using Lagrange's method of variation of parameters, the general solution to \eqref{eq:VRTE-8M} may be expressed as
\begin{equation}
\label{eq:sol-general-nonhomog}
\hat{S}(z) = \Phi(z,z_0)\hat{S}(z_0) + \int_{z_0}^{z}\Phi(z,z')\hat{\sigma}(z')dz'.
\end{equation}
See \citep[\S 2.4]{modern-control-theory} for expanded details.

\subsection{Diagonalising the depth-independent transfer matrix}
\label{sec:diagonalising}

When the transition matrix $K=K(z)$ is independent of $z$ --in homogeneous water-- a fundamental solution to the linear problem $\frac{dF}{dz}(z) = KF(z)$ is given by the exponential series $F(z)=\exp(Kz)$. The transition matrix is consequently of the form
\begin{equation*}
\Phi(z,z_0)=\exp(K(z-z_0)).
\end{equation*}
The key observation here is to note that the depth-independent transfer matrix is diagonalisable. This assumption is based on physical grounds \citep[see][p.\ 451]{mobley-LAW}.

\begin{ass}[Linear independence of eigenvectors]
\label{ass:lin-indep-evals}
The eigenvalues of the depth-independent local transfer matrix $K$ are distinct or, equivalently, the eigenvectors of $K$ are linearly independent.
\end{ass}

Explicitly, by Assumption \ref{ass:lin-indep-evals} we have the decomposition
\begin{equation}
\label{eq:eigenvalue-decomposition}
K = E D E^{-1}
\end{equation}
where $E$ is the (invertible) matrix of eigenvectors of $K$ and $D=\diag(d_1,\ldots,d_{8M})$ is a diagonal matrix whose coefficients $\{d_{i}\}_{1\leq i\leq 8M}$ are the distinct eigenvalues of $K$. With \eqref{eq:eigenvalue-decomposition} we may give more explicit detail to the solution in \eqref{eq:sol-general-nonhomog}.

Let $G=E^{-1}F$. Since $E$ is invertible we may view this identity as a change of basis for the solutions in the fundamental matrix $F$ to those in $G$; indeed, $G$ satisfies the augmented problem $\frac{dG}{dz}(z) = DG(z)$ which is solved by $G(z)=\exp(Dz)$. We may thus express the state-transition matrix as
\begin{equation*}
\Phi(z,z_0)=E\exp(D(z-z_0))E^{-1}
\end{equation*}
where $\exp(D(z-z_0))= \diag(\exp(d_{1}(z-z_0)),\ldots,\exp(d_{8M}(z-z_0)))$ is a diagonal matrix dependant exponentially on the eigenvalues of $K$.

In summary, we have obtained a  explicit solution \eqref{eq:sol-general-nonhomog} for $\hat{S}(z)$ dependant on only two things. Firstly, on determining the eigenvalues and eigenvectors of $K$. This is a problem of computational time roughly proportional to the cube of the matrix order, $(8M)^3$. The goal of the next section is to significantly reduce this time using the structure of the polarised scattering processes. Secondly, on evaluating the integral dependant on $\hat{\sigma}(z)$. This later condition depends on the specific form of $\hat{\sigma}(z)$. If, for example, $\hat{\sigma}(z)=\hat{\sigma}(z)$ is also depth-independent then
\begin{equation*}
\hat{S}(z) = E\exp(D(z-z_0))E^{-1}\hat{S}(z_0) + ED\exp(D(z-z_0))E^{-1}\hat{\sigma}
\end{equation*}
and in the source-free ($\sigma=0$) case simply
\begin{equation*}
\hat{S}(z) = E\exp(D(z-z_0))E^{-1}\hat{S}(z_0).
\end{equation*}

\begin{rem}
\label{rem:beers-law}
The attentive reader might well be wondering about the exponential propagation of the Stokes vector. This depends linearly on the set of functions $z\mapsto \exp(d_{i}z)$. It shall be seen in \S \ref{sec:hemi-sym} that the eigenvalues $d_i$ come in $\pm$ pairs determined by the hemispheres $\Sf^{\pm}$. This taken into account, polarised radiative transfer decays exponentially with depth as predicted by experimental observation.
\end{rem}

\section{Deconstruction of eigenvectors}

The goal of this section is to reduce the dimensional length of $8M$ when computing the matrix of eigenvectors $E$ and eigenvalues $D$, as in \eqref{eq:eigenvalue-decomposition}, of the local transfer matrix $K$. This is achieved by exploiting symmetries in the construction of $K$. In particular, we have two leads to follow:
\begin{enumerate}
\itemspacing
\item
The eigenvalues of $K$ naturally divide into pairs corresponding to their influence in a given hemisphere $\Sf^{\pm}$. We recount this argument from the literature \citep[\S 9.4]{mobley-LAW}.

\item
The structure of $K$ is dependant on that of the $4\times 4$ scattering phase matrices $P$. We consider their form as in \S \ref{sec:phase-matrix} and give reduction formulae for computing their eigenvalues and eigenvectors when they are of part-diagonal type.
This is a novel feature of the transfer of polarised light.
\end{enumerate}

\subsection{Hemispherical symmetries}
\label{sec:hemi-sym}

We exploit the symmetry in $K$, as defined in \eqref{eq:K-def}, given by
\begin{equation}
\label{eq:K-symm}
-K = \mat{}{\bm{1}_{4M}}{\bm{1}_{4M}}{} K \mat{}{\bm{1}_{4M}}{\bm{1}_{4M}}{}.
\end{equation}
Since the right-hand side of \eqref{eq:K-symm} is similar to $K$ it has the same set of eigenvalues, which by this identity must be equal to the eigenvalues $\{-d_i\}$ of $-K$. We conclude that the eigenvalues of $K$ come in positive and negative pairs. Let us re-index the det $\{d_i\}$ accordingly so that
\begin{equation*}
0< d_1 < d_2 < \cdots < d_{4M}
\end{equation*}
and, for $1\leq i \leq 4M$, $d_{i+4M} = -d_{i}$. Note that these inequalities are strict by Assumption \ref{ass:lin-indep-evals}. Defining the partial diagonal matrix $D^{+}:=\diag( d_1,\ldots, d_{4M})$ we have $D=\diag(D^{+},-D^{+})$. Now let us consider a corresponding decomposition of the eigenvectors in $E$. We quarter the array by denoting
\begin{equation*}
E=\mat{E^{++}}{E^{-+}}{E^{+-}}{E^{--}}.
\end{equation*}
By the symmetry in \eqref{eq:K-symm}, this matrix must equal
\begin{equation*}
\mat{}{\bm{1}_{4M}}{\bm{1}_{4M}}{} E \mat{}{\bm{1}_{4M}}{\bm{1}_{4M}}{}= \mat{E^{--}}{E^{+-}}{E^{-+}}{E^{++}}.
\end{equation*}
We implement this identity by denoting $E^{+}:=E^{++}=E^{--}$ and $E^{-}:=E^{-+}=E^{+-}$. There is an analogous block structure for $E^{-1}$ found via the standard inversion formula for $2\times 2$ block matrices;
\begin{equation}
\label{eq:E-inverse}
\mat{{E}^{+}}{{E}^{-}}{{E}^{-}}{{E}^{+}}^{-1}
=\mat{(E^{+}-E^{-}(E^{+})^{-1}E^{-})^{-1}}{(E^{-}-E^{+}(E^{-})^{-1}E^{+})^{-1}}
{(E^{-}-E^{+}(E^{-})^{-1}E^{+})^{-1}}{(E^{+}-E^{-}(E^{+})^{-1}E^{-})^{-1}}
,
\end{equation}
however, we do not use this form until later in \S \ref{sec:analytic-solution-reflectance}. In particular, note that the matrices $E^{\pm}$ and $E^{\pm}(E^{\mp})^{-1}E^{\pm}-E^{\mp}$ are invertible. This follows from the linear independence of the eigenvectors they contain (Assumption \ref{ass:lin-indep-evals}).
Applying the explicit definition of $K$ in $\eqref{eq:K-def}$ to $K=EDE^{-1}$ we obtain the four matrix equations
\begin{equation}
\label{eq:E-decomp-3}
\mat
{\tau E^+ - \rho E^-}
{\rho E^+ - \tau E^-}
{\tau E^- - \rho E^+}
{\rho E^- - \tau E^+}
=
\mat
{E^+D^{+}}
{-E^-D^{+}}
{E^-D^{+}}
{-E^+D^{+}}
\end{equation}
where we have written in shorthand $\tau:=\tau(z;\ell)$ and $\rho := \rho(z;\ell)$.
Adding together and then subtracting the two equations in the first row of \eqref{eq:E-decomp-3} gives us
\begin{equation}
\label{eq:sim-equations}
\begin{cases}
(\tau + \rho)(E^+ - E^-)=  (E^+ + E^-)D^{+}\\
(\tau - \rho)(E^+ + E^-) = (E^+ - E^-)D^{+}
\end{cases}
\end{equation}
which by simultaneous solution implies
\begin{equation*}
(\tau - \rho)(\tau + \rho)(E^+ + E^-) =  (E^+ + E^-)(D^{+})^{2}.
\end{equation*}
This determines a new eigenvalue-eigenvector system of size $4M$ for the matrix $(\tau - \rho)(\tau + \rho)$ whose eigenvalues are given by $\{d_{i}^{2}\}_{1\leq i\leq 4M}$ and whose eigenvectors are the columns of the matrix $$\tilde{E}:=E^+ + E^-.$$ Once $\tilde{E}$ and $\tilde{D}:=(D^{+})^{2}$ have been computed, we recover $D^{+}$, and hence $D$, via the square root and $E^+ - E^-$ via substitution into the second equation. By elimination the individual summands $E^{\pm}$ are also determined. Altogether, the entire $8M$-dimensional system $K=EDE^{-1}$ is reduced to a $4M$-dimensional one, saving factor of $8=2^3$ in computation time.

\subsection{Polarised scattering matrix symmetries}

Now we consider the explicit form of $(\tau - \rho)(\tau + \rho)$ and seek to determine its matrix of eigenvectors $\tilde{E}$ and eigenvalues $\tilde{D}$. By definition we have
\begin{equation*}
\tau\pm \rho =\left[\frac{\frac{2\pi}{N}\Delta\mu_{j}(\hat{P}^{+}_{i,j}\pm \hat{P}^{-}_{i,j}) -c(z)\delta_{i,j}\bm{1}_{4}}{\mu_{i}} \right]_{1\leq i,j\leq M}
\end{equation*}
where we have applied the shorthand $\hat{P}^{\pm}_{i,j}:=\hat{P}^{\pm}(z;\mu_{i},\mu_{j};\ell)$. In this raw form, we firstly make note of the following symmetry for $\hat{P}^{+}_{i,j}\pm \hat{P}^{-}_{i,j}$ to speed computation.

\begin{prop}
Write $P(z;(\mu_{j},0){\rightarrow}(\mu_{i},\phi_{v}))=R(\alpha_{i,j,v})M_{i,j,v}R(\alpha_{i,j,v}')$
as given in \S \ref{sec:phase-matrix}, here explicating the dependence on ${i, j,v}$. Then for $1\leq i,j \leq M$ we have
\begin{multline*}
\hat{P}^{+}_{i,j}\pm \hat{P}^{-}_{i,j} 
=\\
 \sum_{v=0}^{N-1}
\left(R(\alpha_{i, j,v})M_{i, j,v}e\left(\frac{\ell v}{N}\right) \pm R(-\alpha_{i, j,v})M_{i,-j,\frac{N}{2}-v}e\left(\frac{\ell v}{N}-\frac{\ell}{2}\right)\right)
R(\alpha_{i, j,v}').
\end{multline*}
\end{prop}

\begin{proof}
To the definition of $\hat{P}^{+}_{i,j}$ (see \S \ref{sec:deriving-spectral} and \eqref{eq:pm-notation}) we apply Lemma \ref{lem:coords-flip} and factorise the matrix $R(\alpha_{i, j,v}')$ from the Fourier transform.
\end{proof}

\begin{rem}
Working with the Fourier transform here introduces a technical `full stop' in terms of exploiting the block structure of $R(\alpha)MR(\alpha')$. For instance, should one wish to compute the eigenvalues of the matrix $P(z;(\mu',\phi'){\rightarrow}(\mu,\phi))$ then one could conjugate by a permutation matrix; for example,
\begin{equation*}
\begin{bmatrix}
	&&1			\\
	1&&			\\
	&\bm{1}_2&	
\end{bmatrix}
\begin{bmatrix}
	1&&		\\
	&R(\alpha)&		\\
	&&1
\end{bmatrix}
\begin{bmatrix}
	&1&			\\
	&&\bm{1}_2	\\
	1&&	
\end{bmatrix}
=
\begin{bmatrix}
	1&&		\\
	&1&		\\
	&&R(\alpha)
\end{bmatrix}
\end{equation*}
would permit one to instead consider two $2\times 2$ blocks to determine the eigenvalues of $P(z;(\mu',\phi'){\rightarrow}(\mu,\phi))$. However, the benefit in doing so is outweighed by the sacrifice of the decoupling of the azimuthal variable.
\end{rem}

\subsubsection{Part-diagonal type scattering symmetries}

We now assume that scattering matrices $M=\diag(M_1,M_2)$ are of part-diagonal type: without loss in generality, let us assume $M_2$ is a diagonal matrix. Then, each phase scattering matrix decomposes into a $3\times 3$ block and a $1\times 1$ block along the diagonal. In particular we have
\begin{equation*}
(\tau - \rho)(\tau + \rho) 
=\begin{bmatrix}
	A_{ij}	&						\\
				&	\alpha_{ij}
\end{bmatrix}_{1\leq i,j\leq M}
\end{equation*}
where $A_{ij}$ are $3\times 3$ matrices and $\alpha_{ij}$ are scalar coefficients. Consider the matrix $3M\times 3M$ matrix
\begin{equation*}
A:=
\begin{bmatrix}
A_{ij}
\end{bmatrix}_{1\leq i,j\leq M}
\end{equation*}
and the $M\times M$ matrix
\begin{equation*}
\alpha:=
\begin{bmatrix}
\alpha_{ij}
\end{bmatrix}_{1\leq i,j\leq M}
.
\end{equation*}

\begin{prop}
\label{prop:eigenvalue-A-alpha}
The eigenvalues of $(\tau - \rho)(\tau + \rho)$ are given by the eigenvalues of $\diag(A,\alpha)$ which are equal precisely to the eigenvalues of $A$ and of $\alpha$ respectively.
\end{prop}

\begin{proof}
The characteristic polynomial of $(\tau - \rho)(\tau + \rho)$ is given by
\begin{equation*}
\begin{array}{r c >{\ds} l}
\vspace{0.15in}
\det(x\bm{1}_{4M}-(\tau - \rho)(\tau + \rho))&=&\det\mat{\delta_{ij}x\bm{1}_{3}-A_{ij}}{}{}{\delta_{ij}x-\alpha_{ij}}_{1\leq i,j \leq M}
\\\vspace{0.15in}
&=&\pm\det\mat{x\bm{1}_{3M}-A}{}{}{x\bm{1}_{M}-\alpha}
\\
&=&\pm\det(x\bm{1}_{3M}-A)\det(x\bm{1}_{M}-\alpha)
\end{array}
\end{equation*}
where the first equality is obtained by permuting pairs of adjacent rows so that row $4(M-k+1)$ is moved down in order to row $4M-k+1$ for $k=1,\ldots,M$, incurring a factor of $-1$ for each permutation. The zeros of the left-hand side are equal, with multiplicity, to the zeros of the right-hand side. These are precisely the eigenvalues in question. (Note that the sign $\pm$ is equal to $\prod_{k=1}^{M}(-1)^{3(k-1)}=(-1)^{3M(M-1)/2}$ which is determined by the parity of $M(M-1)/2$, which is even if $M\equiv 0$ or $1$ modulo $4$ and odd if $M\equiv 2$ or $3$ modulo $4$.)
\end{proof}

\begin{prop}
\label{prop:eigenvector-A-alpha}
Let $v_A$ and $v_\alpha$ denote right eigenvectors of $A$ and $\alpha$, respectively, and let $\bm{0}_n$ denote the zero column vector of dimension $n\geq 1$.
Then the $4M$-dimensional vectors $(v_A, \bm{0}_M)^{T}$ and $(\bm{0}_{3M}, v_\alpha)^{T}$ are both right eigenvectors of $\diag(A,\alpha)$.
Moreover, the eigenvectors of $(\tau - \rho)(\tau + \rho)$ are given by the eigenvectors of $\diag(A,\alpha)$ after permuting their rows as in the proof of Proposition \ref{prop:eigenvalue-A-alpha}.
\end{prop}

\begin{rem}
In the part-diagonal type case
we have reduced the $8M$-dimensional eigenvalue problem firstly, as in \S \ref{sec:hemi-sym}, to obtain a $4M$-dimensional problem. Then, via Propositions \ref{prop:eigenvalue-A-alpha} and \ref{prop:eigenvector-A-alpha}, to two eigenvalue problems of respective dimension $3M$ and $M$.
The computational time taken to solve such an eigenvalue problem is roughly proportional to the cube of the dimension. We thus save a factor of $$\frac{512M^3}{27M^3+M^3}=\frac{128}{7}> 18.$$
\end{rem}






\section{Infinitely deep inhomogeneous water bodies}
\label{sec:infinite}

For polarised light travelling through a natural plane-parallel water body, boundary conditions are imposed at the surface and at the bottom, constituting  a two-point boundary value problem. At the surface, the reflectance and transmission functions for polarised radiation are considered by \cite{mobley-15}.
At the bottom, there are two separate boundaries to consider:
\begin{itemize}
\itemspacing

\item[($1$)] A finitely deep opaque Lambertian surface.

\item[($\infty$)] An infinitely deep opaque bottom with a homogeneous, source-free layer separating it from the main body.

\end{itemize}

The finite case ($1$) is straightforward to derive \citep[see][\S 1.2.6.1]{mobley-VRTE}. In this case we assume that light scatters diffusely according to Lambert's cosine emission law \citep{lambert}. For example, this is the case for both muddy bottoms and sandy bottoms. Moreover, it is understood that such surfaces are depolarising; the only surviving component of the Stokes vector is that determining the radiance.

Here we derive the bottom boundary condition in case ($\infty$) by computing the `Vector Bidirectional Reflectance Distribution Function' (VBRDF) which determines the polarised radiation reflected back upwards from a finite depth, past which the water is considered homogeneous.

\subsection{The landscape}

Fix two depths $z_1> z_0 >0$ and assume that in the interval the $[z_0,z_1]$ the water body is homogeneous and source-free; recall that this means $K$ is a constant function of $z$ and $\sigma=0$ (see Assumption \ref{ass:homog}). This then describes case ($\infty$) above upon taking the limit $z_1\rightarrow\infty$. Above $z_0$ we make no additional assumptions so that $z_0$ may be seen as the bottom of a general inhomogeneous problem. The information we input is the down-welling radiance distribution $S(z_0;\mu',\phi')$ for $\mu\in [0,1]$ and $\phi\in[0,2\pi)$.

\begin{ass}
\label{ass:reflectance-1}
All up-welling ($\mu<0$) at $z_0$ is due to the reflection of a down-welling incident Stokes vector $S(z_0;\mu',\phi')$ with $\mu'\geq 0$; that is, we impose that $S(z_0;\mu',\phi')=0_{M}$ for all $\mu'<0$.
\end{ass}
For natural water bodies, this assumption describes both the finite ($1$) and infinite ($\infty$) cases above.
We now derive the reflectance $r(z_0,z_1; (\mu',\phi'){\rightarrow}(\mu,\phi))$ at $z_0$ for the homogeneous body $[z_0,z_1]$. By Assumption \ref{ass:reflectance-1}, the boundary condition at $z_0$ may then be written explicitly as
\begin{equation}
\label{eq:boundary-1}
S(z_0; -\mu, \phi) = \int_{0}^{2\pi} \int_{0}^{1} r(z_0,z_1; (\mu',\phi'){\rightarrow}(-\mu,\phi)) S(z_0;\mu',\phi')d\mu'd\phi
\end{equation}
for all $\mu\in [0,1]$ and $\phi\in[0,2\pi)$.
\begin{ass}
\label{ass:reflectance-2}
The reflectance of the water in $[z_0,z_1]$ depends only on the difference between the resultant and incident azimuthal directions $\phi-\phi'$. That is, $$r(z_0,z_1; (\mu',\phi'){\rightarrow}(\mu,\phi))=r(z_0,z_1; (\mu',0){\rightarrow}(\mu,\phi-\phi')).$$
\end{ass}
In the infinitely deep boundary case ($\infty$), this assumption holds given the plane-parallel assumption on the water body itself. (In fact, this holds for the finite boundary cases ($1$) without an imposing asymmetrical structure.)

\subsection{Discrete spectral boundary condition}

In parallel to \S \ref{sec:deriving-spectral}, we derive a discrete spectral analogue for the boundary condition \eqref{eq:boundary-1} to serve as a boundary condition to the discrete spectral VRTE. We apply similarly apply Assumptions \ref{ass:step-az} and \ref{ass:step-po} to $r$, that it is a step function (or quad averaged) in its $\Sf$-variables. Then applying the finite Fourier transform \eqref{eq:fourier-transform} to $r$ we obtain
\begin{equation}
\label{eq:fourier-transform-r}
\hat{r}(z_0,z_1;\mu,\mu';\ell) := \sum_{v=0}^{N-1}r(z_0,z_1;(\mu',0){\rightarrow}(-\mu,\phi_{v}))e(-\ell v/N).
\end{equation}
Substituting 
\eqref{eq:fourier-transform-S}
and
\eqref{eq:fourier-transform-r}
into
\eqref{eq:boundary-1}
via the inversion formula \eqref{eq:fourier-inversion}; evaluating the $\phi'$ integral via
Assumption \ref{ass:reflectance-2}; and equating the $\ell$-th Fourier coefficients using Lemma \ref{lem:equating-coeffs}
we obtain the $4$-dimensional linear equations
\begin{equation}
\label{eq:boundary-2}
\hat{S}^{-}(z_0; \mu_i; \ell) =\frac{2\pi}{N} \sum_{j=1}^{M}\Delta\mu_j \hat{r}(z_0,z_1; \mu_{i},\mu_{j};\ell)\hat{S}^{+}(z_0; \mu_j; \ell)
\end{equation}
for each $1\leq i\leq M$ and $0\leq \ell\leq N-1$. Invoking the matrix notation
\begin{equation*}
\hat{r}(z_0,z_1; \ell):=\left[\hat{r}(z_0,z_1; \mu_{i},\mu_{j};\ell)\right]_{1\leq i,j\leq M},
\end{equation*}
we summarise \eqref{eq:boundary-2} by the single equality
\begin{equation}
\label{eq:boundary-3}
\hat{S}^{-}(z_0; \ell) =\frac{2\pi}{N} \diag(\Delta\mu_1,\ldots,\Delta\mu_{M}) \hat{r}(z_0,z_1; \ell)\hat{S}^{+}(z_0; \ell).
\end{equation}
We now derive a solution for $\hat{r}(z_0,z_1; \ell)$ so that this boundary condition may be utilised.

\subsection{An analytic solution for the reflectance}
\label{sec:analytic-solution-reflectance}

The fundamental assumption of radiative transfer \citep{mobley-LAW} is the (global) \textit{linear interaction principle}: the response radiances emitted from an interval of water $[z_0,z_1]$ depend linearly on the incident radiances on $[z_0,z_1]$. This predicts the existence of four global transfer matrices $T(z_0,z_1; \ell)$, $T(z_1,z_0; \ell)$, $R(z_0,z_1; \ell)$, and $R(z_1,z_0; \ell)$ that satisfy
\begin{equation}
\label{eq:global-interaction-principle}
\begin{bmatrix}
\hat{S}^{+}(z_1; \ell)\\
\hat{S}^{-}(z_0; \ell)
\end{bmatrix}
=
\begin{bmatrix}
T(z_0,z_1; \ell)&R(z_1,z_0; \ell)\\
R(z_0,z_1; \ell)&T(z_1,z_0; \ell)
\end{bmatrix}
\begin{bmatrix}
\hat{S}^{+}(z_0; \ell)\\
\hat{S}^{-}(z_1; \ell)
\end{bmatrix}
\end{equation}
for each $0\leq \ell \leq N-1$. The variable order of the transfer matrices gives the location of the incident radiation in the first argument.

The successful invariant imbedding method \citep{preisendorfer, mobley-LAW} is derived by writing \eqref{eq:global-interaction-principle} in terms of a postulated state-transition matrix $\Phi(z_1,z_0)$, as introduced in \S \ref{sec:fundamental-sol}. The outcome is to rewrite the two-point boundary value problem as two initial value problems, knitted at the seam via the global interaction principle.

This shall too be the goal at present, albeit in a vastly simplified setting. Firstly, by Assumption \ref{ass:reflectance-1}, there is no incident radiance at $z_1$. This meaning that $S^{-}(z_1; \ell)=0_{4M}$. Moreover, $\hat{S}^{+}(z_1; \ell)$ is of no concern to us. Then \eqref{eq:global-interaction-principle} reduces to the boundary condition \eqref{eq:boundary-3} upon writing
\begin{equation}
\label{eq:R-as-r}
R(z_0,z_1; \ell) = \frac{2\pi}{N} \diag(\Delta\mu_1,\ldots,\Delta\mu_{M}) \hat{r}(z_0,z_1; \ell).
\end{equation}
On the other hand, assuming that the water in $[z_0,z_1]$ is source free by \eqref{eq:sol-general-nonhomog} we have that $\hat{S}(z_1)=\Phi(z_1,z_0)\hat{S}(z_0)$. Additionally assuming that the water in $[z_0,z_1]$ is homogeneous then
\begin{equation}
\label{eq:Phi-blocks}
\Phi(z_1,z_0)=
\mat{E^{+}}{E^{-}}{E^{-}}{E^{+}}
\mat{\exp(D^{+}(z-z_0))}{}{}{\exp(-D^{+}(z-z_0))}
\mat{{E}^{+}}{{E}^{-}}{{E}^{-}}{{E}^{+}}^{-1}
\end{equation}
as derived in \S \ref{sec:diagonalising} with the block structure refined in \S \ref{sec:hemi-sym}. To solve for $R(z_0,z_1; \ell)$ in terms of this decomposition of $\Phi(z_1,z_0)$ into the eigenvectors and eigenvalues of $K$, let us introduce the arbitrary notation
$$\Phi(z_1,z_0) = \begin{bmatrix}
\Phi^{++}&\Phi^{-+}\\
\Phi^{+-}&\Phi^{--}
\end{bmatrix}.$$
Solving $\hat{S}(z_1)=\Phi(z_1,z_0)\hat{S}(z_0)$ instead in terms of the response radiances $\hat{S}^{+}(z_1; \ell)$ and
$\hat{S}^{-}(z_1; \ell)$ we obtain
\begin{equation}
\label{eq:invariant-imbedding-reln}
\begin{bmatrix}
\hat{S}^{+}(z_1; \ell)\\
\hat{S}^{-}(z_0; \ell)
\end{bmatrix}
=
\begin{bmatrix}
\Phi^{++}-\Phi^{-+}(\Phi^{--})^{-1}\Phi^{+-} &\Phi^{-+}(\Phi^{--})^{-1}\\
-(\Phi^{--})^{-1}\Phi^{+-} &(\Phi^{--})^{-1}
\end{bmatrix}
\begin{bmatrix}
\hat{S}^{+}(z_0; \ell)\\
\hat{S}^{-}(z_1; \ell)
\end{bmatrix}.
\end{equation}
This equation determines a general \textit{invariant imbedding relation} in more generality than required here. Recall that by Assumption \ref{ass:reflectance-1} we consider $\hat{S}^{-}(z_1; \ell)=0$. Then \eqref{eq:invariant-imbedding-reln} recovers the boundary condition \eqref{eq:boundary-3}, substituting \eqref{eq:R-as-r}, if and only if
\begin{equation*}
R(z_0,z_1; \ell) =- (\Phi^{--})^{-1}\Phi^{+-}.
\end{equation*}
Above, we have assumed that $\Phi^{--}$ is an invertible $4M\times 4M$ block. As we shall imminently see, this follows from that $E^{\pm}$ and $\tilde{E}^{\pm}$ are invertible, which in turn is due to the linear independence of the eigenvectors they contain (see Assumption \ref{ass:lin-indep-evals}). 
Multiplying the components of \eqref{eq:Phi-blocks} and applying the expression for $E^{-1}$ in \eqref{eq:E-inverse} we evaluate $\Phi^{--}$ and $\Phi^{+-}$, thus obtaining the general result for the reflectance at $z_0$ for the interval $[z_0,z_1]$:
\begin{multline}
\label{eq:reflectence-z0-z1}
R(z_0,z_1; \ell) = -
(E^{-}\exp(D^{+}(z_1-z_0))\tilde{E}^{-}+E^{+}\exp(-D^{+}(z_1-z_0))\tilde{E}^{+})^{-1} \\
\times (E^{-}\exp(D^{+}(z_1-z_0))\tilde{E}^{+}+E^{+}\exp(-D^{+}(z_1-z_0))\tilde{E}^{-} ),
\end{multline}
introducing the shorthand $\tilde{E}^{\pm}:=(E^{\pm}-E^{\mp}(E^{\pm})^{-1}E^{\mp})^{-1}$ for the blocks of $E^{-1}$. To evaluate this expression in an infinitely deep water body we must take the limit $R(z_0,\infty; \ell):=\lim_{z_1\rightarrow \infty}R(z_0,z_1; \ell)$. We execute this by looking for the exponential decay of $\exp(-D^{+}z_1)$. We thus rearrange \eqref{eq:reflectence-z0-z1} by inserting the term $1=E^{-}\exp(D^{+}(z_1-z_0))\exp(-D^{+}(z_1-z_0))(E^{-})^{-1}$ in between the two factors to obtain
\begin{multline*}
R(z_0,z_1; \ell) = -
(\tilde{E}^{-}+\exp(-D^{+}(z_1-z_0))(E^{-})^{-1}E^{+}\exp(-D^{+}(z_1-z_0))\tilde{E}^{+})^{-1} \\
\times (\tilde{E}^{+}+\exp(-D^{+}(z_1-z_0))(E^{-})^{-1}E^{+}\exp(-D^{+}(z_1-z_0))\tilde{E}^{-} ).
\end{multline*}
Taking the limit $z_1\rightarrow\infty$ we lastly recover the reflectance of $[z_0,\infty)$ for a homogeneous water body 
\begin{equation*}
\begin{array}{rl}\vspace{0.1in}
R(z_0,\infty; \ell)
=&-(\tilde{E}^{-})^{-1}\tilde{E}^{+}\\\vspace{0.1in}
=& -(E^{-}-E^{+}(E^{-})^{-1}E^{+})
(E^{+}-E^{-}(E^{+})^{-1}E^{-})^{-1}\\\vspace{0.1in}
=&(E^{+}(E^{-})^{-1}-E^{-}(E^{+})^{-1})E^{+}(E^{-})^{-1}(E^{+}(E^{-})^{-1}-E^{-}(E^{+})^{-1})^{-1}\\
=&E^{+}(E^{-})^{-1}.
\end{array}
\end{equation*}

\section{Asymptotic radiance distributions}
\label{sec:asymp}

As noted in \citep[\S 9.6]{mobley-LAW}, experimental observation shows that the scalar radiance distribution depends only on the inherent optical properties in deep, homogeneous source-free waters. In this case, the decay rate of the radiance varies exponentially with depth.
We conclude here that the same holds for each component of the Stokes vector.

As before, consider an infinitely deep homogeneous, source-free interval $[z_0,\infty)$ in the water body. Assume we know the initial surface condition $\hat{S}(z_0)$. In this interval, propagation of the Stokes vector is described by $\hat{S}(z)=\Phi(z,z_0)\hat{S}(z)$ for $z> z_0$ with $\Phi$ as in \eqref{eq:phase-matrix}. Expanding the expression with $E$ as in \S \ref{sec:hemi-sym} we obtain
\begin{equation*}
\begin{bmatrix}
\hat{S}^{+}(z;\ell)\\
\hat{S}^{-}(z;\ell)
\end{bmatrix}
=
\begin{bmatrix}
E^{+}\exp(D^{+}(z-z_0))I^{+} +E^{-}\exp(-D^{+}(z-z_0))I^{-}\\
E^{-}\exp(D^{+}(z-z_0))I^{+} +E^{+}\exp(-D^{+}(z-z_0))I^{-}
\end{bmatrix}
\end{equation*}
where we have defined
\begin{equation}
\label{eq:Ipm}
I^{\pm} := (E^{+}-E^{-}(E^{+})^{-1}E^{-})^{-1}\hat{S}^{\pm}(z_0;\ell) +  (E^{-}-E^{+}(E^{-})^{-1}E^{+})^{-1}\hat{S}^{\mp}(z_0;\ell).
\end{equation}
The main result of the last section stated that $R(z_0,\infty; \ell)=E^{+}(E^{-})^{-1}$. Moreover, the boundary condition \eqref{eq:boundary-3} implies $\hat{S}^{-}(z_0;\ell)=R(z_0,\infty; \ell)\hat{S}^{+}(z_0;\ell)$. Substituting these into \eqref{eq:Ipm} we find that $$I^{+}=0$$ and $$I^{-} = (E^{-}-E^{+}(E^{-})^{-1}E^{+})(\bm{1}_{4M} - R(z_0,\infty; \ell)^{2})\hat{S}^{+}(z_0;\ell),$$ the latter of which does not vanish in general. The result that we have now derived is the asymptotic distribution
\begin{multline}
\label{eq:asymptotic distribution}
\hat{S}^{\pm}(z;\ell) = E^{\mp} \exp(-D^{+}(z-z_0))\\
\times(E^{-}-E^{+}(E^{-})^{-1}E^{+})(\bm{1}_{4M} - R(z_0,\infty; \ell)^{2})\hat{S}^{+}(z_0;\ell)
\end{multline}
as $z\rightarrow\infty$. This equations exhibits the exponential decay of the Stokes vector transform in each variable. One may explicate further by considering the dominant term in \eqref{eq:asymptotic distribution}. This is given by $ \exp(-d_{1}z)$ since $d_{1}<d_{i}$ for all eigenvalues $i\geq 1$. Taking the Fourier expansion \eqref{eq:E-inverse} we similarly obtain a leading exponential term by which the full Stokes vector decays. The polarised Stokes vector components decay exactly as the scalar radiance does in the limit $z\rightarrow\infty$.

\section*{Acknowledgments}

This work was supported by the UK Engineering and Physical Sciences Research Council Prosperity Partnership, TEAM-A (EP/R004781/1).


\bibliography{refs-overSee}			

\end{document}